\documentclass[acmsmall,nonacm]{acmart}

\setcopyright{none}
\acmYear{2027}
\acmDOI{}
\acmISBN{}
\acmConference[ICDCN 2027]{28th International Conference on Distributed
Computing and Networking}{January 4--7, 2027}{NITK Surathkal, India}

\usepackage{amsmath}
\usepackage{array}
\usepackage{tikz}

\theoremstyle{acmplain}
\newtheorem{observation}{Observation}

\begin{document}

\title[Stateless Recovery and Byzantine Quorum Safety]{When Can Stateless Recovery Defeat Byzantine Quorum Safety? A Tight Normal Form for Single-Step BFT}

\author{Arnab Mallick}
\authornote{Corresponding Author.}
\email{arnabm@cdac.in}
\orcid{0009-0008-5458-6429}
\affiliation{%
  \institution{Information Security Group}
  \institution{Centre for Development of Advanced Computing}
  \city{Hyderabad}
  \state{Telangana}
  \country{India}
}

\author{Indraveni Chebolu}
\email{indravenik@cdac.in}
\orcid{0009-0009-6128-6253}
\affiliation{%
  \institution{Information Security Group}
  \institution{Centre for Development of Advanced Computing}
  \city{Hyderabad}
  \state{Telangana}
  \country{India}
}

\begin{abstract}
Byzantine quorum safety relies on correct replicas refusing to sign
conflicting values. A replica that loses its protocol state during recovery
but retains its identity and signing key may forget an earlier vote. We study
certificates formed by matching signed votes from at least $q$ of $n$
replicas, assuming that each correct replica avoids conflicting votes between
recoveries. If two conflicting certificates form, their overlap has size
between $2q-n$ and $b+c$, where $b$ counts Byzantine replicas and $c$ counts
correct identities that recovered during the execution considered. Our main
result decomposes the slack $b+c-(2q-n)$ into four nonnegative counts:
extra signers in the first certificate, extra signers in the second,
identities in neither certificate, and Byzantine or recovering identities
outside their overlap. Zero slack forces an exact signer partition. With
$n=3f+1$ replicas, threshold $q=2f+1$, at most $f$ Byzantine replicas, and
exactly one correct recovery event, any conflicting pair forces exactly $f$
Byzantine replicas, all in the overlap together with the recovered replica;
each certificate has a disjoint side of $f$ correct replicas. A minimal
protocol attains this form.
We distinguish certificate formation from acceptance, give a sufficient check
using configured fault and recovery caps, and explain why durable vote records
written before signature release prevent the conflict.
\end{abstract}



\maketitle

\section{Introduction}

Many Byzantine fault-tolerant (BFT) protocols use signed votes to certify a
decision. In a system of $n=3f+1$ replicas tolerating at most $f$ Byzantine
replicas, two certificates with $2f+1$ signers share at least $f+1$
identities. At least one shared signer is correct. When correct replicas
refuse to sign conflicting values, this overlap prevents conflicting
certificates.

Recovery can invalidate the refusal rule. A replica may lose its vote record
yet retain the identity and key needed to vote again.
We call this \emph{state-erasing recovery}. Even if the replica follows its
rules before and after recovery, its signatures may support both sides of a
conflict. A second signature alone, however, need not supply enough votes
for a second certificate.

We ask: \emph{what signer-set structure is forced when state-erasing
recovery permits conflicting certificates?} We study \emph{single-step}
certificates: any matching valid votes from at least $q$ distinct replicas
constitute a certificate, where $q$ is the certificate threshold. A correct
replica must avoid conflicting votes within each interval between recoveries;
we call this property \emph{incarnation-cleanness}.

The intersection lower bound $2q-n$ is classical. Our contribution is an
exact accounting of the slack between that bound and the actual number of
Byzantine or recovering identities. For any conflicting pair, each unit of
slack is accounted for by extra certificate signers, identities in neither
certificate, or Byzantine or recovering identities outside the intersection.
Zero slack forces minimum-size certificates, coverage of all replicas, and
an intersection containing every Byzantine and recovering identity. With
$n=3f+1$, $q=2f+1$, and one correct recovery event, this forces the partition
in Figure~\ref{fig:quorums}.

\paragraph{Position and scope.}
Classical quorum systems supply the intersection argument~\cite{malkhi1998};
PBFT's recovery treatment preserves the quorum state on which it
depends~\cite{castro2002}. Trusted mechanisms such as A2M, TrInc,
and Aegis constrain equivocation or rollback~\cite{a2m2007,trinc2009,pallas2026}.
We characterize the signer sets of a pair that has formed. In multi-phase
protocols, phase, view, locking, and justification rules impose additional
conditions, so applying our accounting to those protocols requires
protocol-specific analysis. Section~\ref{sec:related} gives the detailed
comparison.

To show that the forced partition is attainable, we use \emph{FirstVote-QC},
a minimal protocol whose only voting restriction is a volatile vote record. Section~\ref{sec:example} collects its rules, it serves as a
tightness construction.

\paragraph{Contributions.}
\begin{itemize}
\item \textbf{C1 (exact accounting).} We prove the slack identity and its
tight and near-tight signer-set consequences.
\item \textbf{C2 (tightness and acceptance).} FirstVote-QC realizes the
zero-slack form. We state separately when a formed pair is accepted and
causes conflicting decisions.
\item \textbf{C3 (checks and boundaries).} We derive a sufficient check from
configured caps, give countermodels separating state erasure, conflicting
votes, certificate formation, and acceptance, and explain how durable vote
records prevent conflicting reuse.
\end{itemize}

\section{Model}
\label{sec:model}

Let $\Pi$ be a fixed set of $n$ replicas with distinct identities. A replica
may act as a \emph{validator}, accepting certificates and deciding values.
A fixed set $B\subseteq\Pi$ is Byzantine and may deviate arbitrarily;
$b=|B|\le f$, where $f$ is the configured fault budget. Each replica $p$
holds a signing key $\mathsf{sk}_p$, and correct replicas' signatures cannot
be forged. Our standard parameters are $n=3f+1$ and $q=2f+1$, with integer
$f\ge1$.

We study one decision instance $d$, such as a block height or log slot.
The notation $x\perp_d y$ means that $x$ and $y$ are incompatible values:
both cannot be decided for this instance. The application supplies this
symmetric, irreflexive conflict relation.

A \emph{vote} $\mathrm{vote}_p(d,x)$ is a signed message supporting $x$ for
$d$. It is \emph{valid} if correctly signed for $(d,x)$; it is
\emph{locally valid} if the signer's protocol rules permit it in its current
state. Correct replicas emit only locally valid votes. Votes conflict when
their values satisfy $x\perp_d y$.

Fix a certificate threshold $1\le q\le n$. A \emph{threshold certificate}
$C_x$ contains valid votes for $(d,x)$ from at least $q$ distinct identities;
its \emph{signer set} is $Q_x:=\mathsf{Signers}(C_x)$. The certificates are
\emph{single-step}: any such set of matching votes is a certificate.
Certificates therefore expose their signer sets; aggregate threshold
signatures that hide signers are excluded. For multi-phase protocols,
additional phase and view checks must be established separately
(Section~\ref{sec:related}).

An \emph{execution prefix} is a finite initial segment of an execution.
It is \emph{admissible} if it respects the fault budget, authentication,
correct replicas' rules, and the recovery semantics below. A prefix
\emph{witnesses} a certificate pair when it contains both certificates and
all votes used to form them. Our results require no synchrony assumption;
the constructions use message delays within finite prefixes.

\begin{definition}[Certificate safety]
No two conflicting threshold certificates $C_x,C_y$ are accepted by correct
validators.
\end{definition}

\begin{observation}
\label{obs:link}
If conflicting $C_x,C_y$ are accepted by distinct correct validators, each
undecided immediately beforehand, and each acceptance causes the validator
to decide its certified value, then consensus safety is violated.
\end{observation}

\paragraph{Stateless recovery.}
A \emph{state-erasing crash} of a correct replica $p$ maps its state as
$(s_p,\mathsf{sk}_p,p)\to(s_0,\mathsf{sk}_p,p)$: protocol state is lost,
identity and key remain, and $s_0$ is the initial protocol state.
We call recovery after this event \emph{state-erasing (stateless) recovery}.
A recovered replica may resume voting, rebuild state first, or remain
retired. For a prefix $E$ witnessing a pair, let
$R_E\subseteq\Pi\setminus B$ be the set of correct identities that undergo
at least one state-erasing recovery in $E$, and let $c=|R_E|$. We suppress
the subscript when the prefix is clear. Repeated recoveries of one identity
count once in $c$; an identity is counted even if it never votes twice.
Recovery is modeled only for otherwise-correct replicas. We use $\uplus$
for disjoint union; in particular, $B\cap R_E=\varnothing$.

\begin{definition}[Incarnation-cleanness]
\label{def:iclean}
An \emph{incarnation} is a correct replica's interval from initialization
or recovery until its next state-erasing crash, or the rest of its execution
if no such crash occurs. A protocol is \emph{incarnation-clean} for instance
$d$ if no correct replica emits locally valid votes for two incompatible
values $x\perp_d y$ during one incarnation.
\end{definition}

Thus conflicting votes from a correct identity require a recovery between
them. Incarnation-cleanness is an explicit voting restriction; consensus
safety alone does not imply it, since additional protocol rules may prevent
conflicting votes from forming certificates or causing decisions.

\section{Why State Erasure Is Not Enough}
\label{sec:example}

Our tightness construction uses a minimal protocol to isolate the effect
of losing a vote record. We analyze safety only.

\paragraph{FirstVote-QC}
Fix a leader for $d$ and set $q=2f+1$. Correct replicas follow these rules:
\begin{enumerate}
\item \textbf{State.} Keep a volatile vote record $\mathit{voted}(d)$ and
a volatile decision record $\mathit{decided}(d)$, both initially $\bot$.
\item \textbf{Vote.} On a leader-signed $\mathrm{PROPOSE}(d,v)$, if
$\mathit{voted}(d)=\bot$, set $\mathit{voted}(d):=v$ before releasing
$\mathrm{vote}_p(d,v)$ to the leader. Ignore later proposals in that
incarnation.
\item \textbf{Certify and decide.} Any $q$ matching votes from distinct
identities form a certificate. An undecided replica accepts the first
certificate it receives after checking its signatures and signer count,
and sets $\mathit{decided}(d)$ to its value, irrespective of its own vote.
\item \textbf{Recover.} On state-erasing recovery, reset both records to
$\bot$. The replica can then vote again with the same identity and key.
\end{enumerate}
The vote record makes FirstVote-QC incarnation-clean.

Without recovery it is safe: two $(2f{+}1)$-certificates intersect in at least
$f+1$ identities, so at least one shared identity is correct. Conflicting
certificates would require that replica to vote for both values, which its
vote record forbids. With one stateless recovery of a correct $r$, the vote
record resets and $r$ can vote again. Section~\ref{sec:theorem} shows that
whenever a one-recovery conflicting pair occurs, its signer sets are \emph{forced} into
Figure~\ref{fig:quorums}.

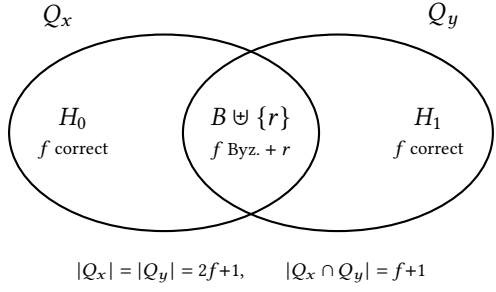
\begin{figure}[t]
\centering
\begin{tikzpicture}
  \draw[thick] (-1.15,0) ellipse (2.05 and 1.3);
  \draw[thick] (1.15,0)  ellipse (2.05 and 1.3);
  \node at (-2.55,1.55) {$Q_x$};
  \node at (2.55,1.55)  {$Q_y$};
  \node[align=center] at (-2.35,0) {$H_0$\\[-1pt]{\scriptsize $f$ correct}};
  \node[align=center] at (0,0)     {$B\uplus\{r\}$\\[-1pt]{\scriptsize $f$ Byz.\ $+\ r$}};
  \node[align=center] at (2.35,0)  {$H_1$\\[-1pt]{\scriptsize $f$ correct}};
  \node[align=center] at (0,-1.85)
      {\scriptsize $|Q_x|=|Q_y|=2f{+}1,\qquad |Q_x\cap Q_y|=f{+}1$};
\end{tikzpicture}
\Description{Two overlapping ellipses labelled Q_x and Q_y. The
left-only region is H_0 (f correct replicas), the right-only region is H_1
(f correct replicas), and the overlap is the disjoint union of B and {r}: the f
Byzantine replicas together with the single recovered replica r.}
\caption{Forced signer partition with one correct recovery event
(Corollary~\ref{cor:nf}). The overlap contains the entire Byzantine set $B$
and the recovered correct replica $r$. The exclusive sets $H_0,H_1$ each
contain $f$ correct replicas that do not recover.}
\label{fig:quorums}
\end{figure}

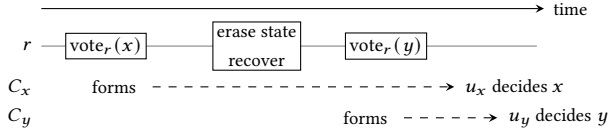
\begin{figure}[t]
\centering
\begin{tikzpicture}[x=0.95cm,y=0.62cm,>=stealth]
  \draw[->] (0,2.45) -- (7.0,2.45) node[right] {\scriptsize time};

  \node[left] at (0,1.65) {\scriptsize $r$};
  \draw[gray] (0,1.65) -- (6.9,1.65);
  \node[draw,fill=white,inner sep=1.5pt] at (0.9,1.65)
      {\scriptsize $\mathrm{vote}_r(x)$};
  \node[draw,fill=white,inner sep=1.5pt,align=center] at (3.0,1.65)
      {\scriptsize erase state\\[-1pt]\scriptsize recover};
  \node[draw,fill=white,inner sep=1.5pt] at (4.8,1.65)
      {\scriptsize $\mathrm{vote}_r(y)$};

  \node[left] at (0,0.78) {\scriptsize $C_x$};
  \node[anchor=east] at (1.45,0.78) {\scriptsize forms};
  \draw[dashed,->] (1.55,0.78) -- (5.75,0.78);
  \node[anchor=west] at (5.78,0.78)
      {\scriptsize $u_x$ decides $x$};

  \node[left] at (0,0.12) {\scriptsize $C_y$};
  \node[anchor=east] at (4.95,0.12) {\scriptsize forms};
  \draw[dashed,->] (5.05,0.12) -- (6.35,0.12);
  \node[anchor=west] at (6.38,0.12)
      {\scriptsize $u_y$ decides $y$};
\end{tikzpicture}
\Description{A time line with replica r voting for x, undergoing the sole
state-erasing recovery, and then voting for y. Certificate C_x forms before
the recovery and is delayed before delivery to u_x, certificate C_y forms
after recovery and is delivered separately to u_y. The two validators decide
x and y, respectively.}
\caption{An attainable fork. Replica $r$ votes for $x$, recovers, then
votes for $y$. Separate delayed deliveries let distinct undecided correct
validators $u_x,u_y$ decide the conflicting values.}
\label{fig:timeline}
\end{figure}

\paragraph{Worked example ($n=7,f=2$).}
Here $q=5$, $b=2$, $c=1$, and $2q-n=b+c=3$. Let
$B=\{b_1,b_2\}$, $H_0=\{h_1,h_2\}$, and
$H_1=\{h_3,h_4\}$, with $r$ correct. Before recovery, $r$ votes for $x$,
so the five signers $Q_x=B\cup H_0\cup\{r\}$ form $C_x$. After $r$ erases
its state and recovers, it votes for $y$, and
$Q_y=B\cup H_1\cup\{r\}$ forms $C_y$. Thus
$Q_x\cap Q_y=B\uplus\{r\}$ has size $3=f+1$, while
$Q_x\cup Q_y=\Pi$. Delaying $C_x,C_y$ and delivering them separately to
undecided $u_x\in H_0,u_y\in H_1$ yields conflicting decisions.

\section{From Formation to Acceptance}
\label{sec:premises}

The counting results concern formation of a conflicting certificate pair.
Formation means that both signed vote sets exist. Acceptance is a separate
protocol event: a validator may impose checks beyond the certificate's
signatures and signer count.

\begin{definition}[Separable acceptance of a pair]
\label{def:sa}
Acceptance is \emph{local} when a validator uses only a certificate and its
own volatile state, with no external reuse oracle. A pair $C_x,C_y$ is
\emph{separably accepted} in prefix $E$ if two distinct, still-undecided
correct validators $u_x,u_y$ accept $C_x,C_y$, respectively, within $E$,
each without receiving the other certificate or its votes.
\end{definition}

In FirstVote-QC, signature and count checks permit the separate deliveries in
Figure~\ref{fig:timeline}.

\section{Bound, Slack Decomposition, and Normal Form}
\label{sec:theorem}

\paragraph{Takeaway.}
Two conflicting certificates share at least $2q-n$ identities. Every shared
identity is either Byzantine or correct and recovered. The slack identity
then accounts exactly for the difference between this lower bound and
the total number of Byzantine or recovering identities in the prefix.

\begin{theorem}[Intersection bound]
\label{thm:bound}
For a single-step protocol that is incarnation-clean at $n$ replicas with
threshold $q$, let an admissible prefix $E$ with actual Byzantine set $B$, $|B|=b$,
and recovering-correct set $R_E$, $|R_E|=c$, contain conflicting threshold
certificates $C_x,C_y$ whose votes all occur in $E$, with signer sets
$Q_x,Q_y$. Then
\[
2q-n\ \le\ |Q_x\cap Q_y|\ \le\ b+c .
\]
Consequently, a conflicting certificate pair requires $b+c\ge 2q-n$.
\end{theorem}
\begin{proof}
Let $I=Q_x\cap Q_y$. Since $|Q_x|,|Q_y|\ge q$ and
$|Q_x\cup Q_y|\le n$,
\[
|I|=|Q_x|+|Q_y|-|Q_x\cup Q_y|\ge 2q-n.
\]
Every identity in $I$ cast valid votes for both conflicting values. A correct
such identity did so in different incarnations by incarnation-cleanness and
therefore belongs to $R_E$; every other identity in $I$ belongs to $B$.
Thus $I\subseteq B\uplus R_E$ and $|I|\le b+c$.
\end{proof}

\begin{corollary}[Zero-recovery certificate safety]
\label{cor:zero}
If $c=0$, a conflicting pair needs $b\ge2q-n$. At
$n=3f+1,q=2f+1$, this requires $b\ge f+1$ and is impossible under $b\le f$.
\end{corollary}

\begin{theorem}[Slack decomposition]
\label{thm:slack}
Under Theorem~\ref{thm:bound}'s conditions, put $I=Q_x\cap Q_y$ and define
\begin{align*}
a_x&:=|Q_x|-q, & a_y&:=|Q_y|-q,\\
u&:=n-|Q_x\cup Q_y|,
& e&:=|(B\uplus R_E)\setminus I|.
\end{align*}
All four quantities are nonnegative integers, and
\begin{equation}
\label{eq:slack}
\underbrace{b+c-(2q-n)}_{\Delta}=a_x+a_y+u+e.
\end{equation}
Hence, for $k\ge0$, if $\Delta\le k$, each certificate has at most $q+k$ signers, their
union contains at least $n-k$ identities, and at most $k$ identities of
$B\uplus R_E$ lie outside their intersection.
\end{theorem}
Here $a_x$ and $a_y$ are the \emph{certificate oversizes}: each counts
signers beyond the required threshold $q$. The term $u$ counts replicas
in neither certificate. The term $e$ counts Byzantine or recovering
identities that do not sign \emph{both} certificates, including those that
sign only one. Equation~\eqref{eq:slack} distributes the slack exactly among these four counts.
\begin{proof}
Theorem~\ref{thm:bound}'s proof gives $I\subseteq B\uplus R_E$. Because
$B\cap R_E=\varnothing$, $|B\uplus R_E|=b+c$; hence $e$ is nonnegative and
$b+c=|I|+e$. Inclusion--exclusion gives
\[
|I|=(q+a_x)+(q+a_y)-(n-u)=2q-n+a_x+a_y+u.
\]
Substitution yields~\eqref{eq:slack}. If the sum is at most $k$, each
nonnegative summand is at most $k$, which gives the three consequences.
\end{proof}

\begin{theorem}[Tight-case normal form]
\label{thm:tight}
Under Theorem~\ref{thm:bound}'s conditions, if $b+c=2q-n$, then, with
$R=R_E$ and $I=Q_x\cap Q_y$,
\[
I=B\uplus R,\qquad |Q_x|=|Q_y|=q,\qquad Q_x\cup Q_y=\Pi.
\]
Moreover, $H_0:=Q_x\setminus I$ and $H_1:=Q_y\setminus I$ are disjoint
correct nonrecovering $(n-q)$-sets, and
$\Pi=(B\uplus R)\uplus H_0\uplus H_1$.
\end{theorem}
\begin{proof}
Here $\Delta=0$. Equation~\eqref{eq:slack} is a sum of nonnegative integers,
so $a_x=a_y=u=e=0$. These equalities give the three displayed conclusions.
The exclusive parts are disjoint, lie outside $B\uplus R$, and have size
$q-|I|=q-(2q-n)=n-q$.
\end{proof}

\begin{corollary}[One-recovery normal form]
\label{cor:nf}
Let the protocol be incarnation-clean and single-step with $n=3f+1$ and
$q=2f+1$. If a prefix containing exactly one state-erasing recovery event by
a correct replica $r$ contains conflicting $C_x,C_y$, then:
(i)~$b=f$, (ii)~$|Q_x|=|Q_y|=2f+1$,
(iii)~$Q_x\cap Q_y=B\uplus\{r\}$, with $r$'s conflicting votes on opposite
sides of its recovery; and (iv)~there are disjoint correct nonrecovering
$f$-sets $H_0,H_1$ such that
$Q_x=B\cup H_0\cup\{r\}$, $Q_y=B\cup H_1\cup\{r\}$, and
$\Pi=B\uplus\{r\}\uplus H_0\uplus H_1$.
\end{corollary}
\begin{proof}
Here $c=1$ and
\[
f+1=2q-n\le|Q_x\cap Q_y|\le b+1\le f+1,
\]
where the final inequality uses $b\le f$.
Thus $b=f$ and $b+c=2q-n$. Theorem~\ref{thm:tight} applies with
$R=\{r\}$. Incarnation-cleanness places $r$'s votes in different
incarnations. Because the prefix contains only that recovery event, the votes
straddle it.
\end{proof}

The normal form determines the partition's shape up to relabeling identities
and exchanging $x$ with $y$. Which identities occupy the two exclusive sets
can vary across executions. Theorem~\ref{thm:slack} bounds deviations when
the slack is positive.

\begin{corollary}[Accepted fork]
\label{cor:acc}
If a conflicting pair $C_x,C_y$ is separably accepted in its admissible prefix,
certificate safety is violated. If both acceptances cause decisions, two
correct validators decide conflicting values (Observation~\ref{obs:link}).
\end{corollary}

\begin{corollary}[Prospective configured-cap check]
\label{cor:caps}
Suppose every admissible prefix satisfies configured bounds
$b\le\bar b$ and $c\le\bar c$. If
$\bar b+\bar c<2q-n$, no conflicting certificate pair can occur.
\end{corollary}
\begin{proof}
Such a pair would require $2q-n\le b+c\le\bar b+\bar c$ by
Theorem~\ref{thm:bound}, a contradiction.
\end{proof}

\noindent\fbox{\parbox{0.94\columnwidth}{\textbf{Audit interpretation.}
The counts $b,c$ describe an execution that has already occurred.
Configured caps $\bar b,\bar c$ give the sufficient check in
Corollary~\ref{cor:caps} before an execution occurs. Counting every recovered
identity is conservative when recovery rules prevent conflicting reuse.}}

\section{Attainability: FirstVote-QC}
\label{sec:inst}

\begin{proposition}
\label{prop:inst}
For $f\ge1$ and $n=3f+1$, every $(d,x,y)$ with $x\perp_d y$ has an admissible
FirstVote-QC prefix containing a conflicting pair in the normal form of
Corollary~\ref{cor:nf}. The pair is separably accepted, and each acceptance
causes its validator to decide the certified value. Thus FirstVote-QC admits
a one-recovery accepted fork.
\end{proposition}
\begin{proof}
Partition $\Pi$ into $B,\{r\},H_0,H_1$ with the sizes in
Corollary~\ref{cor:nf}, and let a Byzantine leader $L\in B$ send
$\mathrm{PROPOSE}(d,x)$ to $r$ and $H_0$, and
$\mathrm{PROPOSE}(d,y)$ to $H_1$. Then $r,H_0$, and $B$ vote for $x$, forming
$C_x$. Replica $r$ undergoes the only recovery, which resets
$\mathit{voted}(d):=\bot$. The leader now sends $y$ to $r$, then $r,H_1$, and
$B$ vote for $y$, forming $C_y$. Every correct vote is locally valid when
cast. Delay both certificates, then deliver $C_x$ to an undecided
$u_x\in H_0$ and $C_y$ to an undecided $u_y\in H_1$, without the opposite
side's certificate or votes. The pair is separably accepted using signature
and signer-count checks, and receipt
causes the two decisions. Since $f\ge1$, both exclusive sets contain a
validator for these deliveries.
\end{proof}

FirstVote-QC's volatile vote record makes the zero-slack form attainable.

\section{Why the Ingredients Are Separate}
\label{sec:indep}

The following countermodels separate three questions: whether recovery
permits a conflicting vote, whether enough votes form a certificate pair,
and whether validators accept that pair.

\begin{proposition}
\label{prop:bif}
State erasure does not imply that a recovered replica can conflict.
\end{proposition}
\begin{proof}
Let a correct replica's vote in $d$ be a fixed function of $d$ and immutable
public data. A crash erases mutable state, but recovery recomputes the same
vote.
\end{proof}

\begin{proposition}
\label{prop:ext}
A recovered replica that can conflict need not enable a conflicting
certificate pair.
\end{proposition}
\begin{proof}
Let $r$ be able to vote $x$ before recovery and $y$ afterward, while every
other correct identity may vote only $x$. The possible $y$-signers satisfy
$S_y\subseteq B\cup\{r\}$ and $|S_y|\le f+1<2f+1=q$, so no
$y$-certificate exists.
\end{proof}

\begin{proposition}
\label{prop:sep}
Formation of a conflicting pair need not yield an accepted fork when
acceptance is not local.
\end{proposition}
\begin{proof}
Keep FirstVote-QC's two vote sets, but equip validators with a linearizable
oracle $\mathsf{Bind}$. The first call $(p,d,v)$ binds signer $p$ to value
$v$ for $d$; subsequent calls succeed exactly when they request the same
value. Acceptance requires successful calls for every signer. Both
certificates may form and either may be accepted alone. But linearizability
orders their calls for the shared signer $r$: once one value is bound, the
other call fails. Thus both certificates cannot be accepted.
\end{proof}

Trusted counters and append-only attestation provide comparable non-reuse
guarantees.

\section{Signer-Side Persistence Suffices}
\label{sec:local}

Preserving safety-relevant state across recovery is a classical requirement.
PBFT's proactive-recovery treatment requires correct replicas to retain
their quorum state; its recovery monitor saves the log, service state, and
checkpoints before reboot~\cite[Sections 5.2.1 and 5.2.3]{castro2002}.
The following proposition gives a sufficient form for our single-step
model: durably record each vote before releasing its signature.

\begin{proposition}
\label{prop:local}
Assume each correct replica durably writes its vote before releasing the
signature, recovery cannot roll back this write, and the replica refuses
every later vote conflicting with that record. Then two conflicting
$(2f{+}1)$-certificates cannot both exist at $n=3f+1$, even if validators
check only signatures and the distinct-signer threshold.
\end{proposition}
\begin{proof}
Two such signer sets intersect in at least
$2(2f+1)-(3f+1)=f+1>|B|$ identities, so a correct identity lies in both. It
would have released conflicting signatures, contrary to the refusal rule
for its durable vote record.
\end{proof}

Durability must precede release. A crash during a write exposes no signature
for that vote. A released signature has a completed record that recovery
cannot erase or roll back. Extending this construction to concurrent
instances of one identity (cloning), outside our model, would need shared
monotonic enforcement of non-equivocation. Validators need no visible proof
of signer persistence within this model.

These mechanisms act at different stages: durable records and safe recovery
prevent conflicting reuse; locks and phase rules constrain certificate
formation; a binding oracle such as $\mathsf{Bind}$
(Proposition~\ref{prop:sep}) constrains acceptance. Blame after
acceptance provides accountability but cannot prevent the accepted fork.

\section{Related Work}
\label{sec:related}

\paragraph{Quorum geometry and conflicting outcomes.}
Byzantine quorum systems impose intersection requirements against
adversary structures~\cite{malkhi1998}. The classical bound $2q-n$ is used
in multi-phase protocols such as PBFT~\cite{castro1999}. Bezerra and
Kuznetsov derive tight multiple-spending bounds in decentralized
cryptocurrencies~\cite{bk2023}, where incompatible outcomes depend on
overlapping support. For a realized conflicting pair, our accounting uses
$Q_x\cap Q_y\subseteq B\uplus R_E$ and decomposes its slack exactly
through~\eqref{eq:slack}.

PBFT uses prepare and commit phases~\cite{castro1999}; Tendermint combines
rounds with locking and valid-value rules~\cite{buchman2018}; HotStuff uses
view-indexed quorum certificates and chained commit rules~\cite{hotstuff2019}.
Their safety arguments additionally depend on these phase rules.

\paragraph{Crash recovery and stable storage.}
Aguilera et al. characterize crash-recovery consensus through failure
detection~\cite{aguilera2000}. Michael et al. recover some shared objects
without local stable storage~\cite{michael2017}, illustrating the dependence
on the object and system assumptions. Wang et al. study Byzantine consensus
with sleepy replicas, identifying storage requirements and resilience
bounds without stable storage under different parameters~\cite{wang2026}.
Those results concern feasibility and construction; ours concern
certificates that have already formed. TEE-Rex studies reliable
storage resilient to rollback under faults weaker than
Byzantine~\cite{teerex2025}; our model combines Byzantine signers with
otherwise-correct signers whose protocol state can be erased.

\paragraph{Mobile faults and cured state.}
Mobile Byzantine models let adversarial influence migrate among processes.
Garay introduces a cured state after an agent leaves a
process~\cite{garay1994}; Biely and Hutle allow every process to be Byzantine
for an interval~\cite{bielyhutle}; trusted-counter variants constrain mobile
faults~\cite{mobilebyz2025}. Our recovering replicas remain otherwise
correct. The count $c$ records correct recovery identities, including
those whose erased vote records permit conflicting signatures; the
Byzantine set $B$ remains fixed.

\paragraph{Trusted anti-equivocation.}
A2M provides attested append-only memory~\cite{a2m2007}, and TrInc uses
trusted monotonic counters to constrain equivocation~\cite{trinc2009}.
Pallas and Aegis address TEE rollback and cloning in blockchain consensus.
Their mechanisms include control over active trusted-component instances
across restart~\cite{pallas2026}. Such mechanisms can enforce the
non-equivocation or non-reuse properties used at the stages distinguished in
Sections~\ref{sec:indep} and~\ref{sec:local}.

\section{Limitations and Conclusion}
\label{sec:conclusion}

The slack identity accounts exactly for deviations from the tight signer
partition. With standard BFT parameters and one correct recovery event,
Corollary~\ref{cor:nf}'s normal form is forced; FirstVote-QC attains it and
admits an accepted fork. Formation sufficiency requires protocol-specific voting
rules beyond the necessary condition $b+c\ge2q-n$. Our assumptions are
fixed identities, attributable single-step certificates, and
incarnation-cleanness. Full formation characterization and multi-phase
recovery analysis remain future work.


\bibliographystyle{ACM-Reference-Format}
\bibliography{icdcn}

@article{malkhi1998,
  author    = {Dahlia Malkhi and Michael Reiter},
  title     = {Byzantine Quorum Systems},
  journal   = {Distributed Computing},
  volume    = {11},
  number    = {4},
  pages     = {203--213},
  year      = {1998},
  publisher = {Springer}
}

@article{aguilera2000,
  author    = {Marcos K. Aguilera and Wei Chen and Sam Toueg},
  title     = {Failure Detection and Consensus in the Crash-Recovery Model},
  journal   = {Distributed Computing},
  volume    = {13},
  number    = {2},
  pages     = {99--125},
  year      = {2000},
  publisher = {Springer}
}

@inproceedings{castro1999,
  author    = {Miguel Castro and Barbara Liskov},
  title     = {Practical {Byzantine} Fault Tolerance},
  booktitle = {Proceedings of the 3rd Symposium on Operating Systems Design
               and Implementation (OSDI)},
  pages     = {173--186},
  year      = {1999},
  publisher = {USENIX Association},
  address   = {Berkeley, CA, USA}
}

@article{castro2002,
  author    = {Miguel Castro and Barbara Liskov},
  title     = {Practical {Byzantine} Fault Tolerance and Proactive Recovery},
  journal   = {ACM Transactions on Computer Systems},
  volume    = {20},
  number    = {4},
  pages     = {398--461},
  year      = {2002},
  doi       = {10.1145/571637.571640},
  publisher = {ACM}
}

@inproceedings{hotstuff2019,
  author    = {Maofan Yin and Dahlia Malkhi and Michael K. Reiter and
               Guy Golan-Gueta and Ittai Abraham},
  title     = {{HotStuff}: {BFT} Consensus with Linearity and Responsiveness},
  booktitle = {Proceedings of the 2019 ACM Symposium on Principles of
               Distributed Computing (PODC)},
  pages     = {347--356},
  year      = {2019},
  publisher = {ACM},
  address   = {New York, NY, USA},
  doi       = {10.1145/3293611.3331591}
}

@inproceedings{michael2017,
  author    = {Ellis Michael and Dan R. K. Ports and Naveen Kr. Sharma and
               Adriana Szekeres},
  title     = {Recovering Shared Objects Without Stable Storage},
  booktitle = {31st International Symposium on Distributed Computing (DISC 2017)},
  series    = {LIPIcs},
  volume    = {91},
  pages     = {36:1--36:16},
  year      = {2017},
  doi       = {10.4230/LIPIcs.DISC.2017.36},
  publisher = {Schloss Dagstuhl--Leibniz-Zentrum f{\"u}r Informatik},
  address   = {Dagstuhl, Germany}
}

@inproceedings{teerex2025,
  author    = {Sadegh Keshavarzi and Gregory Chockler and Alexey Gotsman},
  title     = {{TEE} Is Not a Healer: Rollback-Resistant Reliable Storage},
  booktitle = {39th International Symposium on Distributed Computing (DISC)},
  series    = {LIPIcs},
  volume    = {356},
  pages     = {39:1--39:18},
  year      = {2025},
  doi       = {10.4230/LIPIcs.DISC.2025.39},
  publisher = {Schloss Dagstuhl--Leibniz-Zentrum f{\"u}r Informatik},
  address   = {Dagstuhl, Germany}
}

@inproceedings{wang2026,
  author    = {Chenxu Wang and Sisi Duan and Minghui Xu and Feng Li and
               Xiuzhen Cheng},
  title     = {Consensus in the Known Participation Model with {Byzantine}
               Failures and Sleepy Replicas},
  booktitle = {Network and Distributed System Security Symposium (NDSS)},
  year      = {2026},
  numpages  = {20},
  publisher = {Internet Society},
  address   = {Reston, VA, USA},
  doi       = {10.14722/ndss.2026.230448}
}

@inproceedings{trinc2009,
  author    = {Dave Levin and John R. Douceur and Jacob R. Lorch and
               Thomas Moscibroda},
  title     = {{TrInc}: Small Trusted Hardware for Large Distributed Systems},
  booktitle = {6th USENIX Symposium on Networked Systems Design and
               Implementation (NSDI)},
  pages     = {1--14},
  year      = {2009},
  publisher = {USENIX Association},
  address   = {Berkeley, CA, USA}
}

@inproceedings{a2m2007,
  author    = {Byung-Gon Chun and Petros Maniatis and Scott Shenker and
               John Kubiatowicz},
  title     = {Attested Append-Only Memory: Making Adversaries Stick to Their
               Word},
  booktitle = {Proceedings of the 21st ACM Symposium on Operating Systems
               Principles (SOSP)},
  pages     = {189--204},
  year      = {2007},
  publisher = {ACM},
  address   = {New York, NY, USA}
}

@inproceedings{pallas2026,
  author    = {J{\'e}r{\'e}mie Decouchant and David Kozhaya and Vincent Rahli
               and Jiangshan Yu},
  title     = {Pallas and Aegis: Rollback Resilience in {TEE}-Aided Blockchain
               Consensus},
  booktitle = {Network and Distributed System Security Symposium (NDSS)},
  year      = {2026},
  numpages  = {17},
  doi       = {10.14722/ndss.2026.242443},
  publisher = {Internet Society},
  address   = {Reston, VA, USA}
}

@inproceedings{bk2023,
  author    = {Jo{\~a}o Paulo Bezerra and Petr Kuznetsov},
  title     = {A Tight Bound on Multiple Spending in Decentralized
               Cryptocurrencies},
  booktitle = {27th International Conference on Principles of Distributed
               Systems (OPODIS 2023)},
  series    = {LIPIcs},
  volume    = {286},
  pages     = {31:1--31:19},
  year      = {2024},
  doi       = {10.4230/LIPIcs.OPODIS.2023.31},
  publisher = {Schloss Dagstuhl--Leibniz-Zentrum f{\"u}r Informatik},
  address   = {Dagstuhl, Germany}
}

@inproceedings{garay1994,
  author    = {Juan A. Garay},
  title     = {Reaching (and Maintaining) Agreement in the Presence of Mobile
               Faults (Extended Abstract)},
  booktitle = {Distributed Algorithms, 8th International Workshop (WDAG)},
  series    = {Lecture Notes in Computer Science},
  volume    = {857},
  pages     = {253--264},
  year      = {1994},
  doi       = {10.1007/BFb0020438},
  publisher = {Springer},
  address   = {Berlin, Heidelberg, Germany}
}

@inproceedings{mobilebyz2025,
  author    = {Bo Pan and Maria Potop-Butucaru},
  title     = {Mobile {Byzantine} Agreement in a Trusted World},
  booktitle = {29th International Conference on Principles of Distributed
               Systems (OPODIS 2025)},
  series    = {LIPIcs},
  volume    = {361},
  pages     = {7:1--7:20},
  year      = {2026},
  doi       = {10.4230/LIPIcs.OPODIS.2025.7},
  publisher = {Schloss Dagstuhl--Leibniz-Zentrum f{\"u}r Informatik},
  address   = {Dagstuhl, Germany}
}

@article{bielyhutle,
  author    = {Martin Biely and Martin Hutle},
  title     = {Consensus When All Processes May Be {Byzantine} for Some Time},
  journal   = {Theoretical Computer Science},
  volume    = {412},
  number    = {33},
  pages     = {4260--4272},
  year      = {2011},
  publisher = {Elsevier}
}

@misc{buchman2018,
  author    = {Ethan Buchman and Jae Kwon and Zarko Milosevic},
  title     = {The Latest Gossip on {BFT} Consensus},
  year      = {2019},
  month     = nov,
  eprint    = {1807.04938v3},
  archivePrefix = {arXiv},
  primaryClass  = {cs.DC},
  note      = {Version 3, 22 November 2019},
  url       = {https://arxiv.org/abs/1807.04938v3}
}

\end{document}